\documentclass[runningheads]{llncs}
\usepackage[T1]{fontenc}
\usepackage{graphicx}

\usepackage[utf8]{inputenc}
\usepackage{mathtools, nccmath}

\usepackage{array}
\usepackage{amsmath,amssymb, color, tikz, mathtools}
\usepackage[ruled,vlined]{algorithm2e}
\usepackage{xcolor}
\usepackage{pgfplots}
\usepackage{pgfplotstable}
\usepackage{thmtools}
\usepackage{thm-restate}
\usepackage{comment}
\usepackage{hyperref}

\newcommand{\mathify}[1]{\ifmmode{#1}\else\mbox{$#1$}\fi}
\newcommand{\abs}[1]{\mathify{\left| #1 \right|}}

\renewcommand{\deg}{\mathrm{deg}}

\newcommand{\Adv}{\mathrm{Adv}}
\newcommand{\MM}{\mathrm{MM}}
\newcommand{\CG}{\mathrm{CG}}

\newcommand{\RC}{\mathrm{RC}}
\newcommand{\EC}{\mathrm{EC}}
\newcommand{\FC}{\mathrm{FC}}
\newcommand{\C}{\mathrm{C}}
\newcommand{\SA}{\mathrm{SA}}
\newcommand{\nR}{\mathrm{noisyR}}
\newcommand{\gm}{\mathrm{GapMaj}}

\newcommand{\bs}{\mathrm{bs}}
\newcommand{\fbs}{\mathrm{fbs}}
\newcommand{\s}{\mathrm{s}}

\newcommand{\mc}[1]{\mathcal{#1}}
\newcommand{\mb}[1]{\mathbb{#1}}

\DeclarePairedDelimiter\floor{\lfloor}{\rfloor}
\DeclarePairedDelimiter\mVert{\lVert}{\rVert}
\DeclarePairedDelimiter\pVert{\vert}{\vert}

\begin{document}
\title{On query complexity measures and their relations for symmetric functions}
%
%
\author{Rajat Mittal\inst{1} \and
Sanjay S Nair\inst{1}\orcidID{0009-0008-4187-9935} \and
Sunayana Patro\inst{2}}
\authorrunning{Mittal et al.}
%
\institute{Indian Institute of Technology, Kanpur \\ \email{rmittal@iitk.ac.in} \and
International Institute of Information Technology, Hyderabad }
\maketitle              
\begin{abstract}

The main reason for query model's prominence in complexity theory and quantum computing is the presence of concrete lower bounding techniques: polynomial and adversary method. There have been considerable efforts to give lower bounds using these methods, and to compare/relate them with other measures based on the decision tree.

We explore the value of these lower bounds on quantum query complexity and their relation with other decision tree based complexity measures for the class of symmetric functions, arguably one of the most natural and basic sets of Boolean functions. We show an explicit construction for the dual of the positive adversary method and also of the square root of private coin certificate game complexity for any total symmetric function. This shows that the two values can't be distinguished for any symmetric function. Additionally, we show that the recently introduced measure of spectral sensitivity gives the same value as both positive adversary and approximate degree for every total symmetric Boolean function.

Further, we look at the quantum query complexity of Gap Majority, a partial symmetric function. It has gained importance recently in regard to understanding the composition of randomized query complexity. We characterize the quantum query complexity of Gap Majority and show a lower bound on noisy randomized query complexity (Ben-David and Blais, FOCS 2020) in terms of quantum query complexity.

Finally, we study how large certificate complexity and block sensitivity can be as compared to sensitivity for symmetric functions (even up to constant factors). We show tight separations, i.e., give upper bounds on possible separations and construct functions achieving the same. 

\end{abstract}

\section{Introduction}
\label{ch:intro}

The model of query complexity has been essential in the development of quantum algorithms and their complexity: many of the famous quantum algorithms can be best described in this model~\cite{DBLP:conf/focs/Simon94,DBLP:conf/stoc/Grover96}
and most of the lower bounds known  on complexity of algorithms are obtained through this model~(\cite{beals_quantum_1998,ambainis_quantum_2000}). 

The power of this model for analysing complexity of quantum algorithms arises from the fact that there are concrete mathematical techniques to give lower bounds in this framework. There are two main ways to give lower bounds in this framework. 
\begin{itemize}
    \item Approximate degree and its variants: techniques motivated from capturing the success probability of the algorithm as polynomials~(\cite{beals_quantum_1998}).
    \item Adversary bound: techniques motivated from the adversary method and its semi-definite programming characterization~(\cite{ambainis_quantum_2000,spalek_all_2006,hoyer_negative_2007,lee_quantum_2011,DBLP:conf/stoc/AarBKRT21}).
\end{itemize}

These lower bounds have been motivated by complexity measures of Boolean functions introduced to study deterministic query (decision tree) and randomized query complexity. For deterministic query complexity, measures like Fourier degree, sensitivity, block sensitivity and certificate complexity have been studied extensively~\cite{buhrman_complexity_2002,Nisan1994} (all four are known to lower bound deterministic tree complexity). Similarly randomized certificate complexity is known to be a lower bound randomized query complexity. In last few years many new measures have been introduced to understand these query complexity measures~\cite{ben-david_tight_2020,Chakraborty2022-om}. For example, $\nR$ was introduced to understand the composition of randomized query complexity~\cite{ben-david_tight_2020}, and recently Chakraborty et. al introduced the notion of certificate games whose public coin version is a lower bound on certificate as well as randomized query complexity~\cite{Chakraborty2022-om}. Huang's landmark result~\cite{huang_induced_2019} shows that all these measures are polynomially related to sensitivity.

A natural question: how do these measures relate to each other? Huang~\cite{huang_induced_2019} tells us that these measures are polynomially related. Though, can we figure out what exponent is needed to bound one complexity measure by another (the exponent will depend on these complexity measures)? A lot of work has been done on these relations~\cite{Chakraborty2022-om,DBLP:conf/stoc/AarBKRT21}. (A very nice table with possible separations is given in Aaronson et al.~\cite{DBLP:conf/stoc/AarBKRT21}.)

Let's ask a different question, can we compare these quantities for special class of functions? One of the simplest (and well studied) type of functions are the class of symmetric functions; the output of a symmetric function only depends on the Hamming weight of the input. This class contains many of the natural functions (OR, AND, MAJORITY, PARITY) and has been studied extensively in theoretical computer science. More specifically, related to quantum query complexity, Paturi~\cite{DBLP:conf/stoc/Paturi92} characterized the bounded error approximate degree for any symmetric function. de~Wolf~\cite{DBLP:journals/qic/Wolf10} showed a tight bound for approximate degree with small error by constructing optimal quantum query algorithms. 

The main focus of this work is to examine different lower bound techniques known for quantum query complexity and their relation with other complexity measures for symmetric functions. See the survey by Buhrman and de~Wolf~\cite{buhrman_complexity_2002} for a list of complexity measures based on the query model (we look at these measures when the function is symmetric). For the class of transitive symmetric functions, a study has been initiated by Chakraborty et al.~\cite{Chakraborty2021Separations}.

\subsection{Our results}

For all results in this paper, assume $\epsilon$ to be a constant less than $1/2$.

We have discussed two different lower bounds on bounded error quantum query complexity of a Boolean function: approximate degree and positive adversary.

Our first result shows that for any \emph{total} symmetric function, the positive adversary bound is asymptotically identical to square root of the certificate game complexity~\cite{Chakraborty2022-om}. 
We show it by constructing an explicit solution of the dual of adversary semidefinite program (minimization version, shown in Definition~\ref{MMfprimeDefinition}) which works for the square root of certificate game complexity too.

\begin{theorem}
\label{adversarySolution}
Let $f:\{0,1\}^n \rightarrow \{0,1\}$ be a \emph{total} symmetric Boolean function. 
\begin{equation}
    \Adv^+(f) = \Theta\left(\sqrt{\CG(f)}\right) = \Theta(\sqrt{t_f \cdot n}).
\end{equation}
Here $t_f$ is the minimum $t$ such that $f$ is constant for Hamming weights between $t$ and $n-t$
\end{theorem}

The article~\cite{DBLP:journals/corr/abs-1708-00822} introduced the measure expectational certificate complexity to upper bound Las-Vegas randomized query complexity. It had a very similar optimization program to square root of certificate game complexity (only difference being not having the constraint that weights are less than $1$). This construction implies that bound on weights in the expectational certificate complexity makes it different from square root of certificate game complexity. 

Even though previous results show the value of $\Adv^+(f)$ using the upper bound on $Q_{\epsilon}(f)$, we give an explicit upper bound using the min-max formulation of Adversary bound, which is found to be $\sqrt{t_f\cdot n}$, where $t_f$ is the minimum $t$ such that $f$ is constant for Hamming weights between $t$ and $n-t$. de Wolf~\cite{DBLP:journals/qic/Wolf10} has already shown quantum algorithms with the same quantum query complexity. This shows that $Q_{\epsilon}(f)$ (bounded error quantum query complexity) and $\Adv^+(f)$ are also asymptotically identical to square root of certificate game complexity for \emph{total} symmetric functions. The bound on adversary method was also shown by~\cite{DBLP:journals/toc/AaronsonA14}.

Recently, a lower bound on adversary, called spectral sensitivity, has gained lot of attention and is shown to be a lower bound for approximate degree too~\cite{DBLP:conf/stoc/AarBKRT21}. For any \emph{total} symmetric functions, spectral sensitivity gives the same bound as other two techniques.

\begin{figure}
    \centering
    \includegraphics[width=0.9\textwidth]{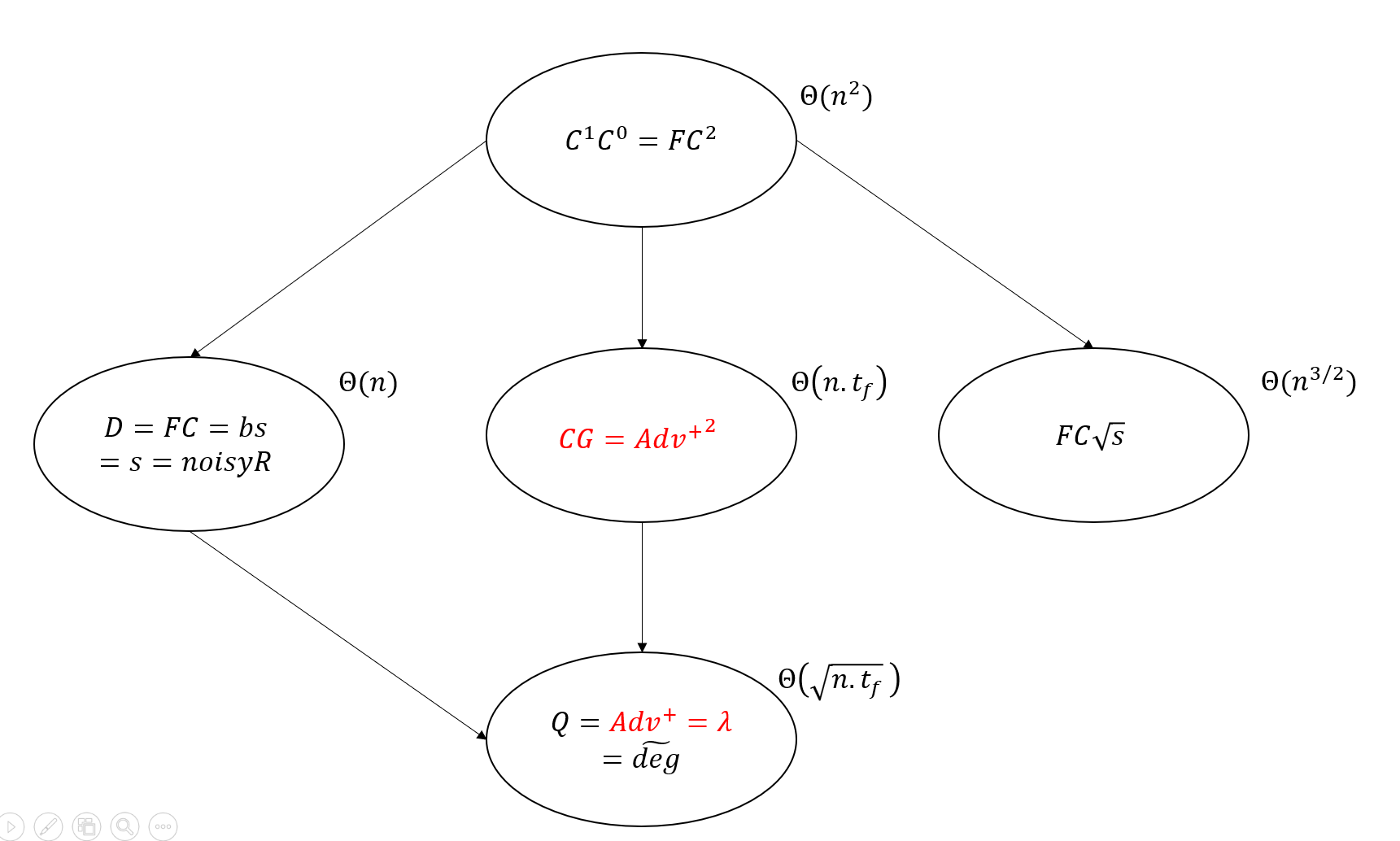}
    \caption{The asymptotic relations on complexity measures of \emph{total} symmetric functions. We show, 
    1) $\sqrt{\CG(f)}$ and $\lambda(f)$ are asymptotically identical to $\Adv^+(f)$; 
    2) $Q_{\epsilon}(f) = O(\nR_{\epsilon}(f)\cdot\sqrt{n})$; 
    3) The relation between $\bs(f)$ ($\C(f)$) and $\s(f)$ are obtained up to constants.
    }
    \label{fig:complexityHierarchyFinal}
\end{figure}

\begin{theorem}
\label{symmetricQueryBounds}
Let $f:\{0,1\}^n \rightarrow \{0,1\}$ be a \emph{total} symmetric Boolean function. Let $\lambda(f)$ denote the spectral sensitivity of $f$ respectively, then
$\lambda(f) = \Theta\left(\sqrt{n\cdot t_f}\right).$

\end{theorem}

This shows that $\lambda(f)$ is identical to the known values of $Q_{\epsilon}(f)$ and $\Adv^+(f)$ for \emph{total} symmetric functions. The approximate degree is also shown to be $\Theta(\sqrt{t_f  \cdot n})$ by Paturi~\cite{DBLP:conf/stoc/Paturi92}. This shows that almost all lower bounds on quantum query complexity in case of symmetric functions give the same value, $\Theta(\sqrt{t_f \cdot n})$.
As far as we know, this only leaves one lower bound for quantum query complexity, known as quantum certificate complexity~\cite{DBLP:conf/coco/Aaronson03,DBLP:journals/cjtcs/KulkarniT16}. It is known that this lower bound is equal to $\Theta(\sqrt{n})$ for any symmetric function.   

Continuing, we examine the quantum query complexity of Gap Majority problem, a partial symmetric Boolean function. This problem gained a lot of attention recently due to the work of Ben-David and Blais~\cite{ben-david_tight_2020} for proving results about composition of randomized query complexity. We prove the following theorem about Gap Majority.

\begin{theorem} \label{gapMajQuery}
Let $\gm_n$ denote the Gap Majority function on $n$ variables and $Q_\epsilon (f)$ denote the quantum query complexity with error $\epsilon$. Then,
\[ Q_{\epsilon} (\gm_n) = \Theta(\sqrt{n}) ,\]
\end{theorem}

We prove the result by giving a quantum algorithm for $\gm_n$ based on quantum counting and proving the tight lower bound using adversary method (lower bound also follows from~\cite{NW99}).  

Ben-David and Blais~\cite{ben-david_tight_2020} recently introduced noisy randomized query complexity (denoted by $\nR_\epsilon(f)$); they showed it to be a  lower bound on the bounded error randomized query complexity (for definition, see~\cite{ben-david_tight_2020}). They also proved that separating noisy randomized query complexity with randomized query complexity is equivalent to giving counterexample for composition of randomized query complexity.

Theorem~\ref{gapMajQuery} allows us to prove a lower bound on the noisy randomized query complexity in terms of quantum query complexity.
\begin{corollary} \label{noisyQu}
Let $\nR_{\epsilon}(f)$ ($Q_\epsilon (f)$) denote the noisy randomized query complexity (quantum query complexity) with error $\epsilon$ respectively. For any \emph{total} Boolean function $f:\{0,1\}^n \rightarrow \{0,1\}$, 
\[ Q_{\epsilon}(f) = O(\nR_{\epsilon}(f)\cdot\sqrt{n}). \] 
\end{corollary}

From the definition of noisy randomized query complexity, it is a lower bound on randomized query complexity. Since quantum query complexity is a lower bound on randomized query complexity, Corollary~\ref{noisyQu} provides a relation between these two lower bounds. All other known lower bounds on randomized query complexity are known to be lower bounds on $\nR$.

We look at the block sensitivity, sensitivity and certificate complexity of \emph{total} symmetric functions too. Since all these measures are $\Theta(n)$ ($n$ - parity) for these functions, we upper bound the separations possible even up to constants. We show that these bounds are tight by constructing functions which achieve these separations. 

\begin{theorem}
\label{C_bs_s_Bounds}
For a symmetric Boolean function $f$, let $\s(f)$,$\bs(f)$ and $C(f)$ denote the sensitivity, block sensitivity and certificate complexity of $f$ respectively.
\begin{enumerate}
    \item $\C(f) \le 2\cdot \s(f)$, there exists $f$ such that $\C(f)= 2 \cdot\s(f)-4$.
    \item $\frac{\bs(f)}{\s(f)} \leq \frac{3}{2}$, there exists $f$ such that $\bs(f) = 3n/4$ and $s(f) = n/2 + 2$.    
\end{enumerate}

\end{theorem}

The known relations on complexity measures for symmetric functions are illustrated in Fig~\ref{fig:complexityHierarchyFinal}.
The preliminaries required for our results are given in Appendix~\ref{ch:chap2}. Proof of Theorem~\ref{adversarySolution} is detailed in Section~\ref{ch:chap3}.  The results about Gap Majority and its consequence are given in Section~\ref{ch:chap4}. The bound on spectral sensitivity is proven in Appendix~\ref{spec_sens_symm}. The proof of Theorem~\ref{C_bs_s_Bounds} is provided in the appendix~\ref{ch:chap5}.

\section{Lower bounds on quantum query complexity for \emph{total} symmetric functions}
\label{ch:chap3}

In this section we first construct an optimal solution for the min-max formulation of the adversary bound. It turns out that a similar construction gives an optimal bound on the private coin version of certificate game complexity.



\paragraph{\texorpdfstring{$\Adv^+(f)$}{} for \emph{total} symmetric Boolean functions:} 
~~ \\

For a Boolean function $f$, we define $t_f$ to be the minimum $t$ such that $f$ is constant between $t$ and $n-t$. We use the min-max formulation of Adversary bound ($\MM(f)$) and explicitly show that $\Adv^+(f) = O(\sqrt{t_{f}\cdot n})$. 


\label{sec:adv_upper}
\begin{theorem}
\label{ExplicitSoln}
For any \emph{total} symmetric Boolean function f, $\Adv^+(f) = O(\sqrt{t_{f}\cdot n})$.
\end{theorem}

\begin{figure}
	\centering
	\includegraphics[width=0.6\textwidth]{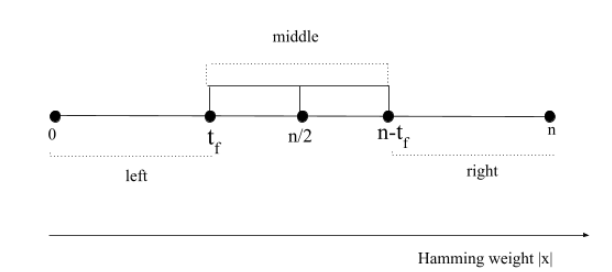}
	\caption{A general \emph{total} symmetric Boolean function viewed on Hamming weights, it is constant in the range $[t_f, n-t_f]$ by definition.}
	\label{f:advExplicit}       
\end{figure}

\begin{proof}

Using $t_f$, $\abs{x}$ can fall in 3 regions; left ($L$), right ($R$)  and middle ($M$). From the definition of $\MM(f)$ given in Appendix~\ref{PositiveAdv}, we need to define a weight function $w$.

\begin{definition}
\label{MMWeightDef}
    We define a weight function w(x,i):
    \begin{itemize}
     \item For $L$ ($\abs{x} < t_f)$:
 \begin{center}
 $w(x,i)$ = $\begin{cases}
    $$\sqrt{n/t_{f}}$$, & \text{if $x_i = 1$ }\\
    $$\sqrt{t_{f}/n}$$, & \text{if $x_i = 0$.}\\
    \end{cases}$
 \end{center} 
 
 \item For $R$ ($\abs{x} > n - t_f)$ :
 \begin{center}
 $w(x,i)$ = $\begin{cases}
    $$\sqrt{n/t_{f}}$$, & \text{if $x_i = 0$ }\\
    $$\sqrt{t_{f}/n}$$, & \text{if $x_i = 1$.}\\
    \end{cases}$
 \end{center}
 
 \item For $M$ ($t_f \leq \abs{x} \leq n-t_f )$ : 
\begin{center}
    
     For any $t_f$ $i$'s such that  $x_i = 1, w(x,i) = \sqrt{n/t_{f}}$. \\
     For any $t_f$ $i'$s such that  $x_i = 0, w(x,i) = \sqrt{n/t_{f}} $.
 \end{center}
\end{itemize}

\end{definition}

For the constraint $\sum_{i : x_{i} \neq y_{i}} \sqrt{w(x,i)\cdot w(y,i)} \geq 1 $ for all $x,y : f(x) \neq f(y)$, following cases arise:
\begin{itemize}

\item $x\in L$ and $y \in R$: For any $x \in L$ and $y \in R$, there are at least $(n-2t_{f}+2)$ indices such that $x_i = 0$ and $y_i = 1$.
  \begin{align*}
           \sum_{i : x_{i} \neq y_{i}} \sqrt{w(x,i)\cdot w(y,i)} & \geq (n-2t_{f}+2)\sqrt{t_{f}/n}
      \\& \geq 1 &&\text{(true for any $0 < t_f \leq n/2$).}
      \end{align*}
      
\item $x \in L$ and $y \in M$; $x \in R$ and $y \in M$; $x,y \in L$; $x,y \in R$:
      We know that for any $y \in M$, there are at least $t_f$ indices such that $y_i = 1$. We also know that for any $x \in L$, the maximum number of indices such that $x_i = 1$ is $<t_f$. Then $\exists $ at least one $i$ such that $x_i = 0$ and $y_i = 1$. For this index $i$, we know that $w(x,i) = \sqrt{t_{f}/n}$ and $w(y,i) = \sqrt{n/t_{f}}$. Thus $\sum_{i : x_{i} \neq y_{i}} \sqrt{w(x,i)\cdot w(y,i)} \geq 1 $ for all $x \in L$ and $y \in M$.
      Similar argument holds for $x \in R$  and $y \in M$  or $x,y \in L$ or $x,y \in R $.
      
\item $x,y \in M$: There are no $x,y \in M$ such that $f(x) \neq f(y)$.

\end{itemize}

With $w$ as the weight scheme, we can see that the value of $\max_{x} \Sigma_{i} w(x,i) $ becomes $\Theta(\sqrt{t_{f} \cdot n})$. Thus, $\Adv(f) = \MM(f) = O(\sqrt{t_{f} \cdot n})$. 
Since we already from Theorem~\ref{symmetricQueryBounds} know that $\Adv(f) = \Theta(\sqrt{t_{f} \cdot n})$, this weight scheme $w$ is an explicit solution for the same.

\end{proof}

\paragraph{Certificate game complexity for symmetric functions:}
~~\\


The article~\cite{Chakraborty2022-om} defined certificate game complexity in various settings. The definition in case of private coin setting is very similar to the min-max definition of the adversary method (Definition~\ref{MMfDefinition}). In particular, the optimization program for the square root of certificate game complexity does not have a square root in the constraints.


\begin{definition}
\label{MMfprimeDefinition}
Let $f:{S \to\{0,1\}}$ where $S \subseteq \{0,1\}^n $ be a Boolean function.
Let $w$ be a weight function, then 

\begin{equation}
\begin{split}
    \sqrt{\CG(f)} = 
    \min_{w} \max_{x \in \text{Dom(f)}} \Sigma_{i \in [n]} w(x,i) \\
    \text{s.t} \sum_{i : x_{i} \neq y_{i}} w(x,i)w(y,i) \geq 1, & \text{ \; $\forall x, y $ : $f(x) \neq f(y)$} \\
         w(x,i) \geq 0, & \text{ \; $\forall x \in Dom(f), i \in [n]$}.
\end{split}
\end{equation}
\end{definition}

 It turns out that the same bound can be obtained for $\sqrt{\CG(f)}$ (with the same explicit solution as the one for the adversary bound). Together with the result of previous section, we get Theorem~\ref{adversarySolution}.

\begin{lemma}
\label{mmPrimeLemma}
    For any \emph{total} symmetric Boolean function f, $\sqrt{\CG(f)} = O(\sqrt{t_{f}\cdot n})$. 
\end{lemma}
\begin{proof}

We consider the same weighing scheme $w$ used for $\MM(f)$. For the \emph{total} symmetric function f where $t_f$ is the minimum value such that the function value is constant for $(t_f,n-t_f)$, we use a weight function $w(x)$ as defined in Definition~\ref{MMWeightDef}

First, we verify that $w$ satisfies $\sum_{i : x_{i} \neq y_{i}} w(x,i)w(y,i) \geq 1$ for all $x, y  : f(x) \neq f(y)$. A case analysis, similar to proof of Theorem~\ref{ExplicitSoln}, verifies the constraint (notice that the square root does not matter because the proof only requires at least one index such that $w(x,i) w(y,i) \geq 1$ ).


      
      

The objective value $\max_{x} \Sigma_{i} w(x,i) $ remains $\Theta(\sqrt{t_{f} \cdot n})$. Thus, $\sqrt{\CG(f)} = O(\sqrt{t_{f} \cdot n})$. 

\end{proof}

\begin{proof}[Proof of Theorem~\ref{adversarySolution}]
Suppose $f$ is symmetric.
Since it is known that $\sqrt{\CG(f)}$ is lower bounded by $\Adv^+(f)$, Lemma~\ref{mmPrimeLemma} shows that $\sqrt{\CG(f)}=\Theta(\Adv^+(f))$. The lower bound on $\Adv^+(f)$ follows from~\cite{DBLP:journals/toc/AaronsonA14}.
\end{proof}



It is easy to see that the other certificate game complexity measures, $\CG^{pub}(f) = \CG^*(f) = \CG^{ns}(f)$ are $\Theta(n)$ for any total symmetric function $f$~\cite{Chakraborty2022-om}.

Expectational certificate complexity ($\EC(f)$)~\cite{DBLP:journals/corr/abs-1708-00822}, looks misleadingly similar to $\MM(f)$. 

Since $\EC(f) \geq \FC(f)$ (\cite{DBLP:journals/corr/abs-1708-00822}[Lemma 7]), $\EC(f) = \Omega(n)$ for all symmetric functions $f$. Also Theorem~\ref{ExplicitSoln} shows that $\MM(f)=O(\sqrt{t_{f}\cdot n})$. There are two differences between $\EC(f)$ and $\MM(f)$: objective function and the restriction on the weight scheme. Removing the square root in the objective function doesn't change the upper bound(Lemma~\ref{mmPrimeLemma}). Though, restricting $w(x,i)$ to be in between $0$ and $1$ changes the bounds drastically. In particular, we know that there are lot of symmetric functions for which $\EC$ is asymptotically bigger than $\sqrt{\CG}$.


 \section{Quantum query complexity for Gap majority function}
\label{ch:chap4}
The Gap Majority function (Definition~\ref{defGM}) has been recently used by Ben-David and Blais~\cite{ben-david_tight_2020} to understand the composition of randomized query complexity. They used it to characterize the noisy randomized query complexity, a lower bound on randomized query complexity, and were able to show multiple results on composition using this noisy version. As a first step towards understanding the quantum query complexity of partial symmetric functions, we compute the quantum query complexity of $\gm_n$ and use it show new lower bounds on noisy randomized query complexity.

The main result of this section, Theorem~\ref{gapMajQuery}, shows that $Q_\epsilon(\gm_n) = \Theta(\sqrt{n})$. As a corollary, we obtain $Q_\epsilon(f) = O(\nR_{\epsilon}(f)\cdot\sqrt{n})$ for any total Boolean function. It is not known to be true for partial Boolean functions. 

\begin{proof}[Proof of Corollary~\ref{noisyQu}]
Ben-David and Blais~\cite[Theorem 4]{ben-david_tight_2020} showed that for any Boolean function $f$, 
\begin{equation}
\begin{split}
\label{eqn:noisy}
    R_\epsilon(f \circ \gm_n) = \Theta(\nR_\epsilon(f)\cdot n) \\  \Rightarrow Q_\epsilon(f \circ \gm_n) = O(\nR_\epsilon(f)\cdot n).
\end{split}
\end{equation}

Lee et al.~\cite[Theorem 1.1, Lemma 5.2]{lee_quantum_2011} showed that for any total Boolean function $f$, 
\begin{equation}
\label{eqn:composition}
   Q_\epsilon(f) \cdot Q_\epsilon(\gm_n) = \Theta(Q_\epsilon(f \circ \gm_n)).
\end{equation}

We prove in Theorem~\ref{gapMajQuery} that 
\begin{equation}
    Q_\epsilon(\gm_n) = \Theta(\sqrt{n}).
\end{equation}

Combining the result from Theorem~\ref{gapMajQuery} and (Eq~\ref{eqn:noisy}) and (Eq~\ref{eqn:composition}), we get the required result for any total Boolean function $f$, 
\begin{equation}
\label{eqn:newbound}
    Q_\epsilon(f) = O(\nR_\epsilon(f)\cdot\sqrt{n}).
\end{equation}

\end{proof}



Comparing our bound with previously known results, Ben-David et al.~\cite[Lemma 38]{ben-david_tight_2020} show that $\nR_\epsilon(f) = \Omega(\RC(f))$ (notice $\fbs(f) = \Theta(\RC(f))$~\cite{tal_properties_2013}). Looking at previously known bounds on $\RC$ using quantum query complexity~\cite[Table 1]{DBLP:conf/stoc/AarBKRT21}, we know that $Q_\epsilon(f) = O(\RC(f)^3)$. The best possible bound on $\nR$ in terms of $Q_{\epsilon}$ becomes

\begin{equation}
\label{eqn:existingNoisyBound}
    Q_\epsilon(f) = O(\nR_\epsilon(f)^3).
\end{equation}

Corollary~\ref{noisyQu} gives a better bound than the existing bound when $\nR_\epsilon(f) = \Omega(n^\frac{1}{4})$.

\subsection{Proof of Theorem~\ref{gapMajQuery}}

We start by showing that there exists a quantum algorithm that can compute the quantum query complexity for Gap Majority in $\Theta(\sqrt{n})$ steps, thus giving us the upper bound for $Q_\epsilon(\gm_n)$.

The main tool is the following lemma from Aaronson and Rall~\cite[Theorem 1]{DBLP:conf/soda/AaronsonR20} to estimate the Hamming weight of an input (a modification of quantum approximate counting by Brassard et al.~\cite[Theorem 15]{Brassard_2002}).

\begin{lemma}[Restatement of Theorem 1 from~\cite{DBLP:conf/soda/AaronsonR20}] 
\label{AppCount}
Let $\epsilon > 0$ and $x \in \{0,1\}^n$ be the input whose Hamming weight we want to estimate, and $t$ be the actual Hamming weight of $x$. Given query access to an input oracle for $x$ and an allowed error rate $\delta > 0$, there exists a quantum algorithm that outputs an estimate $t'$ with probability at least $1- \epsilon$ satisfying 
\begin{equation*}
     (1-\delta)t \leq t' \leq (1+\delta)t.
 \end{equation*}
 The above algorithm uses $O(\frac{1}{\delta}\sqrt{n/t})$ queries where the constant depends on $\epsilon$.
\end{lemma}
The upper bound is a straightforward implication of the previous lemma.

\begin{lemma}
\label{upperBoundGapMaj}
$Q_\epsilon(\gm_n)$ = $O(\sqrt{n})$.
\end{lemma}
\begin{proof}
We use approximate counting (Lemma~\ref{AppCount}) in the following way. For $\gm_n$, we know that $t = n/2 \pm \sqrt{n}$, thus we need to choose a $\delta$ such that the minimum estimate of $n/2 + \sqrt{n}$ is greater than the maximum estimate of $n/2 - \sqrt{n}$. When $\delta = 1/\sqrt{n}$, the minimum estimate for $t = n/2 + \sqrt{n}$ and maximum estimate for $t = n/2 - \sqrt{n}$ is $n/2 \pm \sqrt{n}/2 - 1$ respectively and hence there is no overlap. 

Given $t = n/2 \pm \sqrt{n}$ and choosing $\delta = 1/\sqrt{n}$, the quantum algorithm from Lemma~\ref{AppCount} can estimate a non-overlapping $t'$ using $O(\sqrt{n})$ queries with probability at least $1-\epsilon$. Thus $Q_\epsilon(\gm_n)$ = $O(\sqrt{n})$. 
\end{proof}

The matching lower bound was given by~\cite{DBLP:journals/toc/AaronsonA14} using the positive adversary method. We give a complete proof for the sake of completeness.

\begin{lemma}
\label{lowerBoundGapMaj}
$Q_\epsilon(\gm_n)$ = $\Omega(\sqrt{n})$.
\end{lemma}

\begin{proof}

 There are multiple ways of obtaining this lower bound namely using lemma 29 of~\cite{DBLP:journals/toc/AaronsonA14} or by applying theorem 1.1 from \cite{DBLP:conf/stoc/NayakW99} on the function $\gm(f)$. However, we use the positive adversary bound from the original Ambainis article~\cite{ambainis_quantum_2000}.
Using the notation of Theorem~\ref{posAdvTheorem}, let $X$ be the set of all inputs such that $\gm_n(x) = 0$ and $Y$ be the set of all inputs such that $\gm_n(x) = 1$. We take $R$ to be the set of all pairs $(x,y)$ such that the bits which are set to $1$ in $x$ are a subset of the bits which are set to $1$ in $y$. 

For any $x \in X$, the number of $y \in Y$ such that $(x,y) \in R $ are $\binom{\frac{n}{2} + \sqrt{n}}{2\sqrt{n}}$. To enumerate y, we have to look at the number of ways in which we can fill $\frac{n}{2} + \sqrt{n}$ places with $2\sqrt{n}$ ones and $\frac{n}{2} - \sqrt{n}$ zeroes. This is because the rest of the $\frac{n}{2} - \sqrt{n}$ places are the ones which correspond to the ones in $x$. Since this is true for any $x \in X$, the value of $m$ is $\binom{\frac{n}{2} + \sqrt{n}}{2\sqrt{n}}$. Similar argument holds true for $m'$ as well, whose value turns out to be $\binom{\frac{n}{2} + \sqrt{n}}{\frac{n}{2} - \sqrt{n}}$. 

For a particular $i$, if $x_{i} = 0$, the number of $y \in Y$ such that $(x,y) \in R$ and $y_{i} = 1$ is  $\binom{\frac{n}{2} + \sqrt{n}-1}{2\sqrt{n}-1}$. This is because we are fixing the value of $y_{i}$ to be one, and we already have $\frac{n}{2} - \sqrt{n}$ ones filled out from the set bits of x,  so we are left with $2\sqrt{n}-1$ ones to be filled in $\frac{n}{2} + \sqrt{n}-1$ places. For a particular $i$, if $x_{i} = 1$, the number of $y \in Y$ such that $(x,y) \in R$ and $y_{i} = 0$ is $0$ because otherwise this pair $(x,y) \not\in R$.
So, for an $x \in X$, the maximum number of $y \in Y$ such that $(x,y) \in R$ and $x_{i} \neq y_{i}$ is  $\binom{\frac{n}{2} + \sqrt{n}-1}{2\sqrt{n}-1}$. Since this is true for any $x \in X$, the value of $l$ is also the same. 
Similar argument holds true for $l'$ as well, which equals $\binom{\frac{n}{2} + \sqrt{n}-1}{\frac{n}{2} - \sqrt{n}}$.

Substituting the values of $m, m', l, l'$, we get $Q_\epsilon(\gm_n) = \Omega(\sqrt{n})$.

\end{proof}

\begin{proof}[Proof of Theorem~\ref{gapMajQuery}]

From Lemma~\ref{upperBoundGapMaj} and Lemma~\ref{lowerBoundGapMaj}, we can conclude that $Q_\epsilon(\gm_n)$ = $\Theta(\sqrt{n})$.

\end{proof}

We examine the block sensitivity, randomized certificate complexity and $\MM(f)$ of $\gm_n$ in appendix~\ref{sec:gapMajappendix}.


\subsubsection{Acknowledgements} 
We would like to thank Sourav Chakraborty, Manaswi Paraashar and Swagato Sanyal for the discussions.


%
%
%
\bibliographystyle{splncs04}
\bibliography{Refer}

\newpage
\appendix
\section{Preliminaries}
\label{ch:chap2}

\emph{Norm:} For a vector $v$, we use $\mVert*{v}$ to denote its 2-norm.  The spectral norm of a square matrix $\Gamma$ is defined as
\[ \mVert*{\Gamma}=\max_{v:\mVert*{v}=1}\mVert*{\Gamma.v}= \max_{\mVert*{u}=\mVert*{v} =1} u^T \Gamma v . \]

\begin{lemma}\label{Sum of matrices}
    For any two non-negative $n\times m$ matrices $A$ and $B$, $\mVert*{A+B}\ge\max\{\mVert*{A},\mVert*{B}\}$
\end{lemma}

\begin{proof}

Let us assume that $\mVert*{A}\ge\mVert*{B}$, without loss of generality. Let $u,v$ be unit vectors such that $u^TA.v= \mVert*{A}$. Let $\tilde{v},\tilde{u}$ be unit vectors such that $\tilde{v}_i=\pVert*{v_i},\tilde{u}_i=\pVert*{u_i}$. Since $A$ and $B$ are both non-negative, ${(A+B)}_{i,j}\ge A_{i,j} \ge 0, \forall i,j \in [n]$. So, 

\begin{equation}
    \tilde{u}^T.(A+B).\tilde{v} = \sum_i\sum_j\tilde{u}_i.(A+B)_{ij}.\tilde{v}_j \ge \sum_i\sum_j\tilde{u}_i.A_{ij}.\tilde{v}_j \ge \sum_i\sum_jv_i.A_{ij}.v_j = \mVert*{A}.
\end{equation}

This implies,

\begin{equation}
    \mVert*{(A+B)} = \max_{\mVert*{u}=\mVert*{v}=1} u^T.(A+B).v \ge \max_{\mVert*{u}=\mVert*{v}=1} u^T.A.v = \mVert*{A}.
\end{equation}

\end{proof}

\emph{Boolean functions:} A function $f:\mc{D} \rightarrow \{0,1\}$, where $\mc{D} \subseteq \{0,1\}^n$, is called a Boolean function. It is \emph{total} if $\mc{D} = \{0,1\}^n$. Otherwise, if $\mc{D}$ is a strict subset of $\{0,1\}^n$, then it is called a \emph{partial} function. 

The Adjacency matrix of a function $f$ : ${\{0,1\}}^n\to{\{0,1\}}$ is defined as an $n\times n$ matrix $A_f \in \{0,1\}^{n\times n}$, where $A_f[x,y]=1$ if and only if $f(x)\ne f(y)$ and the Hamming distance between $x$ and $y$ is 1.

The Hamming weight $\abs{x}$ of an input $x \in \{0,1\}^n$ is defined as the number of bits which are set to 1. 
A Boolean function $f:\{0,1\}^n \rightarrow \{0,1\}$ is called symmetric if the value of the function only depends on the Hamming weight of the input. Equivalently, for any permutation $\sigma \in S_n$,
\[ f(x_1, x_2, \cdots, x_n) = f(x_{\sigma(1)}, x_{\sigma(2)}, \cdots, x_{\sigma(n)}) .\]
For any symmetric function $f$, we define $t_f$ to be the minimum value such that the function $f$ is constant for Hamming weights between $t_f$ and $n-t_f$. Notice that $t_f \leq n/2$.

Many natural and well studied functions like OR, AND, MAJORITY, PARITY are symmetric. One of the important partial symmetric functions is the Gap Majority.
\begin{definition} \label{defGM}
The Gap Majority function on $n$ variables, called $\gm_n$, is the partial symmetric function 
\begin{equation*}
    \gm_n(x)=\begin{cases}
    0, & \textit{if $\abs{x} = n/2 - \sqrt{n}$ }\\
    1, & \textit{if $\abs{x} = n/2 + \sqrt{n}$}\\
    \textit{not defined}, & \textit{otherwise}.
  \end{cases}
\end{equation*}
\end{definition}

The bounded error quantum query complexity of a Boolean function $f$, called $Q_{\epsilon}(f)$, is the minimum number of queries needed to compute $f$ with error $\epsilon$. 
By repeating the algorithm constant number of times, the success probability can be made $1- \epsilon$ for any constant $0< \epsilon < 1/2$. 
We introduce a few lower bounds on quantum query complexity in the following subsections. For the exact definition and more details about quantum query complexity, 
please look at the survey by Hoyer and Spalek~\cite{DBLP:journals/eatcs/HoyerS05}.

\subsection{Positive adversary}
\label{PositiveAdv}
Ambainis~\cite{ambainis_quantum_2000} introduced the first positive adversary bound, denoted by $\Adv^+(f)$. Later, many modification of it were used to give lower bounds on different problems~\cite{DBLP:conf/coco/BarnumSS03,DBLP:conf/focs/Ambainis03,DBLP:journals/tcs/Zhang05,DBLP:conf/coco/LaplanteM04};
all of those were shown to be equivalent~\cite{spalek_all_2006}. (These methods do not include the generalized (negative) adversary method~\cite{hoyer_negative_2007}.)
We give definitions of a few versions of positive adversary method that will be used in this article.

The following version is from the original article by Ambainis~\cite[Theorem2]{ambainis_quantum_2000}.
\begin{theorem}
\label{posAdvTheorem}
Let $f(x_{1},...,x_{n})$ be a function of n $\{0,1\}$-valued variables and X,Y be two sets of inputs such
that $f(x)\neq f(y)$ if $x \in X$ and $y \in Y$. Let $R \subset X\times Y$ be such that
\begin{enumerate}
    \item For every $x \in X$, there exist at least $m$ different $y \in Y$ such that $(x,y) \in R$.
    \item  For every $y\in Y$ , there exist at least $m'$ different $x\in X$ such that $(x,y) \in R$.
    \item  For every $x \in X$ and $i \in \{1,...,n\}$, there are at most $l$ different $y \in Y$ such that $(x,y) \in R$ and $x_{i} \neq y_{i}$.
    \item For every $y \in Y$ and $i \in \{1,...,n\}$, there are at most $l'$ different $x \in X$ such that $(x,y) \in R$ and $x_{i} \neq y_{i}$.
    
\end{enumerate}
Then, any quantum algorithm computing f uses $\Omega(\sqrt{\frac{mm'}{ll'}})$ queries.
\end{theorem}

The next version is called spectral adversary~\cite{DBLP:conf/coco/BarnumSS03}.
\begin{definition}
Let $f:{\{0,1\}}^n\to\{0,1\}$ be a Boolean function. Let $D_i $, for all $i\in [n]$, be  $2^n\times2^n$ Boolean matrices, where indexes for rows and columns are from inputs ${\{0,1\}}^n$. The $x,y$ entry of matrix $D_i$ is $1$ iff $x_i\ne y_i$. Similarly, $F$ is a $2^n \times 2^n$ Boolean matrix such that $F[x,y]=1 \Leftrightarrow f(x)\ne f(y)$. Let $\Gamma$ be a   $2^n\times2^n$ non-negative symmetric matrix, then  

\begin{equation}
    \SA(f) = \max_{\Gamma: \Gamma\circ F=\Gamma} \frac{\mVert*{\Gamma}} {\max_{i\in[n]}\mVert*{\Gamma\circ D_i}}.
\end{equation}

\end{definition}

Another version is called the minimax adversary method $\MM(f)$~\cite{DBLP:conf/coco/LaplanteM04}, and is a minimization problem. 
\begin{definition}
\label{MMfDefinition}
Let $f:{S \to\{0,1\}}$ where $S \subseteq \{0,1\}^n $ be a Boolean function.
Let $w$ be a weight function, then 

\begin{equation}
\begin{split}
    \MM(f) = 
    \min_{w} \max_{x \in \text{Dom(f)}} \Sigma_{i \in [n]} w(x,i) \\
    \text{s.t} \sum_{i : x_{i} \neq y_{i}} \sqrt{w(x,i)w(y,i)} \geq 1, & \text{ \; $\forall x, y $ : $f(x) \neq f(y)$} \\
         w(x,i) \geq 0, & \text{ \; $\forall x \in Dom(f), i \in [n]$}.
\end{split}
\end{equation}
\end{definition}

We know that $\Adv(f) = \MM(f)= \SA(f) = O(Q_{\epsilon}(f))$~\cite{spalek_all_2006,ambainis_quantum_2000,DBLP:conf/coco/BarnumSS03,DBLP:conf/coco/LaplanteM04}. 

\subsection{Spectral Sensitivity \texorpdfstring{$\lambda(f)$}{}}

In 2020, Aaronson et al. introduced a new measure based on the sensitivity graph of a Boolean function, which can be used to estimate complexity of the function~\cite{DBLP:conf/stoc/AarBKRT21}.  

\begin{definition}
For a total Boolean function $f:{\{0,1\}}^n\rightarrow\{0,1\}$, the spectral sensitivity is defined as the spectral norm of its adjacency matrix $A_f$.

\begin{equation}
    \lambda(f) = \mVert*{A_f} = \max_{v:\mVert*{v}=1} \mVert*{A_f.v}
\end{equation}
\end{definition}

This spectral relaxation of sensitivity was found to be a lower bound for spectral adversary method by Aaronson et al.~\cite{DBLP:conf/stoc/AarBKRT21}. It was also observed that since $G_f$ is symmetric and bipartite, the spectral norm of $A_f$ is simply the largest eigenvalue of $A_f$.

\subsection{Approximate degree}

A multivariate polynomial $p:\mb{R}^n \rightarrow \mb{R}$ is said to approximate a Boolean function $f$ with error $\epsilon$ if
\begin{equation}
    |p(x)-f(x)|\le\epsilon, \forall x\in{\{0,1\}}^n  .
\end{equation}

\begin{definition}
The $\epsilon$-approximate degree of $f:\{0,1\}^n \rightarrow \mb{R}$, $\deg_\epsilon(f)$, is the minimum degree of a polynomial which approximates $f$, i.e.,

\begin{equation}
    \deg_\epsilon(f) = \min_{p:|p(x)-f(x)|\le\epsilon,\forall x\in\{0,1\}^n} \deg(p).
\end{equation}
\end{definition}

We know $Q_{\epsilon}(f) \geq \frac{\deg_\epsilon (f)}{2}$~\cite{beals_quantum_1998}.

\subsection{Sensitivity Measures and Certificate Complexity}
\label{sensitivity}
\begin{table}
    \centering
    \begin{tabular}{|c|c|c|c|}
         \hline
         \textbf{Measure} & \textbf{Sensitivity} & \textbf{Block Sensitivity} & \textbf{Certificate Complexity} \\ \hline
         Local at input $x$ & $\s(f,x)$ & $\bs(f,x)$ & $\C(f,x)$ \\ \hline
         For output $z$ & $\s_z(f)=\max_{x:f(x)=z}\s(f,x)$ & $\bs_z(f)=\max_{x:f(x)=z}\bs(f,x)$ & $\C_z(f)=\max_{x:f(x)=z}\C(f,x)$ \\ \hline
         For function $f$ & $\s(f)=\max_{x}\s(f,x)$ & $\bs(f)=\max_{x}\bs(f,x)$ & $\C(f)=\max_{x}\C(f,x)$ \\ \hline
    \end{tabular}
    \caption{Details of Sensitivity Measures}
    \label{tab:sensitivityDetails}
\end{table}
For a Boolean function $f: {\{0,1\}}^n\to{\{0,1\}}$, every input $x$ is a string of $n$ bits. For an index $i\in[n]$ (or a block of indices $B\subseteq[n]$), define $x^i$ ($x^B$) to be the input where the $i$-th bit (all bits in block $B$) is flipped. An index $i$ (or a block $B$) is called \textit{sensitive} for input $x$ if $f(x) \neq f(x^i)$ ($f(x)\ne f(x^B)$). 

The local sensitivity $\s(f,x)$ (local block sensitivity $\bs(f,x)$) at an input $x$ is the number of sensitive indices (the maximum number of disjoint sensitive blocks) possible in the input $x$.

For any input $x$, a \textit{certificate} is a set of indices $C\subseteq[n]$ such that for any input $y\in{\{0,1\}}^n$, if $x_i=y_i ~~ \forall i\in C$, then $f(x)=f(y)$. The smallest size possible for a certificate at input $x$ is its local certificate complexity $\C(f,x)$. 

The local sensitivity, block sensitivity and certificate complexity can be used to define these measures on a particular output of a function and for a function itself in general.
The precise definitions are given in \ref{tab:sensitivityDetails}.

We know that $\s(f) \leq \bs(f) \leq \C(f)$~\cite{buhrman_complexity_2002}. For any symmetric function $f$, it is easy to see that $\s(f) = \C(f) = \Theta(n)$ where $n$ is the arity of $f$. This implies that all intermediate measures like block sensitivity and fractional certificate complexity are also $\Theta(n)$. 

\subsection{Expectational certificate complexity}

A new complexity measure was introduced by Gavinsky et al~\cite{DBLP:journals/corr/abs-1708-00822} called expectational certificate complexity, defined as follows.

\begin{definition}
\label{ECDefinition}
Let $f:{S \to\{0,1\}}$ where $S \subseteq \{0,1\}^n $ be a Boolean function.
Let $w$ be a weight function, then

\begin{equation}
\begin{split}
    \EC(f)=
    \min_{w} \max_{x \in \text{Dom(f)}} \Sigma_{i \in [n]} w(x,i) \\
    \text{s.t} \sum_{i : x_{i} \neq y_{i}} w(x,i)w(y,i) \geq 1, & \text{ \; $\forall x, y $ : $f(x) \neq f(y)$} \\
    0\leq w(x,i) \leq 1, & \text{ \; $\forall x \in Dom(f), i \in [n]$}.
\end{split}
\end{equation}

\end{definition}

It was shown that $\EC(f) \geq \FC(f)$ (\cite{DBLP:journals/corr/abs-1708-00822}[Lemma 7]) and $\EC(f) = \Omega(n)$ for all symmetric functions $f$.

\section{Appendix}
\label{appendix}

\subsection{Spectral sensitivity of symmetric functions}
\label{spec_sens_symm}

Here, we show that the spectral sensitivity of a \emph{total} symmetric function $f:\{0,1\}^n \rightarrow \{0,1\}$ is  $\Theta(\sqrt{t_f \cdot  n})$ (Theorem~\ref{symmetricQueryBounds}). Remember that for any symmetric Boolean function $f$ on $n$ variables, we define $t_f$ to be the minimum value such that the function $f$ is constant for Hamming weights between $t_f$ and $n-t_f$.

The  spectral sensitivity of a function is given by the spectral norm of the sensitivity graph of the function. First, we will find the spectral sensitivity of threshold functions. Then, we will express the sensitivity graph of a general symmetric functions in terms of the sensitivity graph of threshold functions and obtain tight lower bound on spectral sensitivity.

\subsubsection{Threshold Functions}

A \emph{threshold function} with threshold $k$, $T_k:{\{0,1\}}^n \to \{0,1\}$, is a symmetric Boolean function defined as  
\[
 T_k(x)= \begin{cases}
1 \ \ \  \textit{if }\pVert*{x}\ge k, \\
0 \ \ \  \textit{otherwise}.
\end{cases}
\]

\begin{theorem} \label{thSpSBounds}
    For the threshold function $T_k:{\{0,1\}}^n\to\{0,1\}$ with threshold $k$, 
    \[ \lambda(T_k)=\sqrt{k \cdot (n+1-k)}. \]
\end{theorem}

\begin{proof}

The adjacency matrix of the sensitivity graph of $T_k$ is denoted by $A_{T_k}$. Remember that
\[ \lambda(T_k) = \mVert*{A_{T_k}} = \max_{v:\mVert*{v}=1}\mVert*{A_{T_k} \cdot v} .\]

For any $l$, define $v_l$ with indices in ${\{0,1\}}^n$ to be, 
\[ v_l(x)= \begin{cases}
1 \ \ \ \textit{if } |x|=l  \\
0 \ \ \ \textit{otherwise.}
\end{cases}
\]

The length of $v_l$ is 
\[ \mVert*{v_l} = \sqrt{\sum_{x:\pVert*{x}=l}1^2 + \sum_{x:\pVert*{x}\ne l}0^2} = \sqrt{\binom{n}{l}}. \] 

To prove the lower bound, we will show that $v_k$ achieves a stretch of $\sqrt{k(n+k-1)}$. 
Expanding the value at any index $x$: 
\[ (A_{T_k} \cdot v_k)[x] = \sum_{0\le |y| \le n}A_{T_k}[x,y] \cdot v_k[y]. \]

Since $v_k[y]=1\Leftrightarrow\pVert*{j}=k$, 
\[ (A_{T_k} \cdot v_k)[x] = \sum_{|y|=k} A_{T_k}[x,y] \cdot v_k[y]. \]

Notice that $A_{T_k}[x,y]=1$ and $\pVert*{y}=k$ then $\pVert*{x}=(k-1)$. This implies that $(A_{T_k} \cdot v_k)[x] = 0$ if $\pVert*{x} \neq k-1$.
Also, for any $x: \pVert*{x}=(k-1)$, there are $(n+1-k)$ possible $y$'s such that Hamming distance between $x$ and $y$ is 1 and $\pVert*{y}=k$. So, 

\begin{equation}
    (A_{T_k} \cdot v_k) [x] = (n+1-k)   \ \ \  \textit{if } \pVert*{x} = k-1.
\end{equation}

In other words, $A_{T_k} \cdot v_k=(n+1-k) \cdot v_{k-1}$. Hence, the stretch in the length of vector $v_k$ when multiplied with adjacency matrix $A_{T_k}$ is 

\begin{equation}
    \frac{|A_{T_k} \cdot v_k|}{|v_k|} = \frac{(n+1-k)\cdot\pVert*{v_{k-1}}}{\pVert*{v_k}} = (n+1-k)\cdot\sqrt{\frac{\binom{n}{k-1}}{\binom{n}{k}}} = \sqrt{k\cdot(n+1-k)}.
\end{equation}

To prove the upper bound, we will use the result by Aaronson et al.~\cite{DBLP:conf/stoc/AarBKRT21}, 
\[ \lambda(f) \leq \sqrt{\s_0(f)\cdot\s_1(f)}. \]
For $T_k$, the sensitivity of an input $x$ is $k$ if $|x|=k$ and $n+1-k$ if $|x|=k-1$ (it is $0$ everywhere else). We get the required upper bound by noticing that $\s_0(T_k)=n+1-k$ and $\s_1(T_k)=k$. Since the same lower bound has already been proved,  
\[ \lambda(T_k) = \sqrt{k\cdot(n+1-k)}. \]

\end{proof}

On plotting the spectral sensitivity against the threshold $k$, we see that the spectral sensitivity is minimum when $k$ is $1$(OR) or $n$(AND), and maximum when $k=\frac{n}{2}$ (MAJORITY).

\subsubsection{\emph{Total} symmetric functions}

We observe that the sensitivity graph of any symmetric function $f$ can be written as sum of the sensitivity graphs of a subset of threshold functions. Define $S_f = \{1 \le k \le n: f(k)\ne f(k-1) \}$.

\begin{lemma} \label{expressTh}
    For a symmetric function $f:\{0,1\}^n \rightarrow \{0,1\}$, the adjacency matrix of the sensitivity graph of $f$ can be written as
    \[ A_f = \sum_{S_f}  A_{T_k} .\]
\end{lemma}

\begin{proof}

Let $B=\sum_{S_f} A_{T_k}$. Since the support of $A_{T_k}$ where $k\in S_f$ is disjoint, $B$ is also a $\{0,1\}$ matrix. We need to prove that $B = A_f$.

From the definition of sensitivity graph, $A_f[x,y]=1$ if and only if the Hamming distance between $x$ and $y$ is 1 and $f(x)\ne f(y)$. Without loss of generality ($A_f$ and $B$ are symmetric), assume $\pVert*{x} > \pVert*{y}$, then $\pVert*{x} \in S_f$ implying $B[x,y] = 1$.
 
For the reverse direction, if $B[x,y] = 1$ and $\pVert*{x} > \pVert*{y}$, then $\pVert*{x} \in S_f$. This means $f(x) \neq f(y)$, and the Hamming distance between $x$ and $y$ has to be $1$ from the definition of $A_{T_k}$. Again, $A_f[x,y] = 1$.





\end{proof}

Consider a symmetric function $f:{\{0,1\}}^n\to\{0,1\}$. As defined earlier, $t_f$ is the smallest value such that function value is a constant for the range of Hamming weights $\{t_f,..,n-t_f\}$. We capture the spectral sensitivity of $f$ using $t_f$.


\begin{proof} [Proof of Theorem~\ref{symmetricQueryBounds}]

From Lemma~\ref{expressTh}, the adjacency matrix of the sensitivity graph of $f$ can be written as $A_f = \sum_{S_f} A_{T_k}$. Since each $A_{T_k}$ has only non negative values, Lemma~\ref{Sum of matrices} gives us
\[ \lambda(f) = \mVert*{A_f} \ge \mVert*{A_{T_k}} = \lambda(T_k) ~~~\forall k \in S_f . \]. 

There is a change in function value of $f$ at Hamming weight $t_f$ or $(n+1-t_f)$ and from Theorem~\ref{thSpSBounds}
\[ \lambda(T_{t_f})= \lambda(T_{n+1-t_f}) = \sqrt{t_f.(n+1-t_f)}. \]

We get $\lambda(f) \ge \sqrt{t_f.(n+1-t_f)}$. Since $(n+1-t_f)=\Theta(n)$,
\[ \lambda(f) = \Omega(\sqrt{t_f \cdot n}) . \]

The upper bound follows from $\Adv^+(f)$, Theorem~\ref{ExplicitSoln}.
\end{proof}

\subsection{Sensitivity, block sensitivity and certificate complexity for symmetric functions}
\label{ch:chap5}

For a symmetric function, all three measures (sensitivity, block sensitivity, certificate complexity) are $\Theta(n)$. In this section we explore the separations possible between these measures even within a constant factor. For this section assume $f$ is symmetric unless stated. 

Consider a symmetric function $f:{\{0,1\}}^n\to\{0,1\}$. We can define $f$ to be a function over Hamming weights i.e., $f(w)\in\{0,1\}$ , $w\in [n]$. Any property of an input, like local sensitivity, local block sensitivity and certificate complexity can also be expressed as a property of its Hamming weight. In this section, the input $w$ for $f$ will be considered as a Hamming weight $\in \{0,1, \cdots, n$ and \emph{not} a Boolean string $\in \{0,1\}^n$.

If the symmetric function $f$ is clear from the context, for any Hamming weight $z$, we denote $a_z$ and $b_z$ such that $a_z \le z \le b_z$ and for any Hamming weight $w$ such that $a_z \le w \le b_z$, $f(w)=f(z)$. Also, the value of $(b_z-a_z)$ is the maximum possible. In other words, we find the largest contiguous set of Hamming weights which include $z$ such that the function value remains the same.

\subsubsection{Sensitivity and Certificate Complexity}

First, we show the first part of Theorem~\ref{C_bs_s_Bounds}.

\begin{theorem}
    For a symmetric function $f:{\{0,1\}}^n\to\{0,1\}$, $\C(f) \le 2\cdot \s(f)$.
\end{theorem}

\begin{proof}
Let us consider some Hamming weight $z$. We can find the corresponding $a_z$ and $b_z$. Here, we can see that sensitivity $\s(f,a_z) = a_z$ and $\s(f,b_z) = n-b_z$.
For any input of Hamming weight $z$, we can define a certificate with $a_z$ $1$'s and $(n-b_z)$ $0$'s. With this certificate, only inputs with Hamming weights at least $a_z$ and at most $b_z$ are accepted, and all those have same function value as $z$. So $\C(f,z) = a_z+n-b_z = \s(f,a_z)+\s(f,b_z)$. Since $f(a_z)=f(b_z)=f(z)$, we can say that:

\begin{equation}
    \C_{f(z)}(f) \le \s(f,a_z)+\s(f,b_z) \le 2\cdot\s_{f(z)}(f).
\end{equation}

Since we know, that $\C(f)=max\{\C_0(f),\C_1(f)\}$ and  $\s(f)=max\{\s_0(f),\s_1(f)\}$, the above result can be extended to say that:

\begin{equation}
    \C(f)\le 2\cdot\s(f).
\end{equation}

\end{proof}

We construct a function which achieves this separation.

\begin{lemma}
    There exists a symmetric Boolean function $f$ such that $\C_1(f)=2\cdot\s_1(f)$, and $\C(f)=2\cdot\s(f)-4$.
\end{lemma}

\begin{proof}
We can define a function as follows: $f:{\{0,1\}}^n\to\{0,1\}$, $f(x) = 1$ if $\pVert*{x} \in \{\frac{n-1}{2},\frac{n+1}{2}\}$, and $0$ everywhere else. For this function, we shall see that $\s_0(f)=\C_0(f)=\frac{n+3}{2}$, $\s_1(f)=\frac{n-1}{2}$ $\C_1(f)=n-1$. So, here we have $\C(f)=\C_1(f)$ and $\s(f)=\s_0(f)$. So, $\C_1(f)=2\cdot\s_1(f)$ and $\C(f)=2\cdot\s(f)-4$.
\end{proof}

\subsubsection{Sensitivity and Block sensitivity} \label{sec:s_bs}

Now we show the second part of Theorem~\ref{C_bs_s_Bounds}.

We will first show that $\bs(f)$ is bigger than $\s(f)$ by at most a factor of $3/2$ for any total symmetric function $f$. Additionally, the proof will output a function where this separation is achieved. 

The main result of this section is the following theorem.
\begin{theorem}
    \label{bsfFinalTheorem}
    For any total symmetric Boolean function f, $\frac{\bs(f)}{\s(f)} \leq \frac{3}{2}$.  
\end{theorem}

The proof of this theorem will require multiple results to calculate the block sensitivity for a particular hamming weight, and narrow down the search for the function which optimises the above ratio. First, we obtain a general formula to calculate the local block sensitivity at any input of given Hamming weight $z$, $\bs(f,z)$, in terms of $a_z$ and $b_z$.

\begin{lemma}
\label{blockSensitivityFormula}
    For any total symmetric Boolean function f, for any input with Hamming weight $z \in [n]$ and their corresponding $a_z$ and $b_z$ values, $\bs(f,z)$ can be characterised as follows:
    \begin{equation*}
    \bs(f,z)=\begin{cases}
    \left \lfloor \frac{z}{z-a_z+1} \right \rfloor + \left \lfloor \frac{n-z}{b_z-z+1} \right \rfloor, & \text{if $a_z \neq 0$ and $b_z \neq n$  }\\
    \left \lfloor \frac{n-z}{b_z-z+1} \right \rfloor, & \text{if $a_z = 0$ and $b_z \neq n$  }\\
    \left \lfloor \frac{z}{z-a_z+1} \right \rfloor, & \text{if $a_z \neq 0$ and $b_z = n$  }\\
    0, & \text{if $a_z = 0$ and $b_z = n$, (constant function)}\\
  \end{cases}
\end{equation*}
where $a_z$ and $b_z$ are defined as above. 
\end{lemma}

To prove this, we shall first see that for any given symmetric function $f$, we can achieve block sensitivity at any input $x$ using blocks of only $1$'s or only $0$'s.

\begin{lemma}
\label{oneTypeBlocksOnly}
    For a total symmetric function $f:\{0,1\}^n \rightarrow \{0,1\}$ and any input $x\in{\{0,1\}}^n$, define $l=\bs(f,x)$. There exist sensitive blocks $B_1,B_2...,B_l$ in $x$ such that, 
    \[ \forall k\in [l],~\forall i,j\in B_k,~~~ x_i=x_j. \]
\end{lemma}

\begin{proof}
If there are two  indices $i$ and $j$ in a sensitive block $B$ such that $x_i=1$ and $x_j=0$, we can replace this with another block $B'=B-\{i,j\}$. $B'$ is also sensitive because the Hamming weight of the input caused by flipping the bits in $B'$ is $\pVert*{x^{B'}}=\pVert*{x^B}+1-1=\pVert*{x^B}$. So, $f(x^{B'})=f(x^B)$. We can continue this till we get a block $\widetilde{B}$ such that for $i,j\in\widetilde{B}, x_i=x_j$. Since $\widetilde{B}\subseteq B$ and $f(x^{\widetilde{B}})=f(x^B)$, we can replace $B$ with $\widetilde{B}$. So, for any input $x$, we can have $\bs(f,x)$ disjoint sensitive blocks, such that for any block $B$ in this, for $i,j\in B$, $x_i=x_j$.
\end{proof}

\begin{proof}[Proof of Lemma~\ref{blockSensitivityFormula}]
\label{bsFormulaProof}

For any input of Hamming weight $z\in[a_z,b_z]$, a block of 1s of size $(z-a_z+1)$ is a sensitive block, as flipping those 1s will give an input with the Hamming weight to $(a_z-1)$ which will have a different function value. Similarly, any block of 0s of size $(b_z-z+1)$ is also a sensitive block. We know from  Lemma~\ref{oneTypeBlocksOnly} that for this input, we can define $\bs(f,z)$ disjoint sensitive blocks such that the value of input is same at any two indices in a block. Suppose there is a block $B$ where input value is 1 at all the indices. Flipping indices of this block will give a new input with Hamming weight $z-\pVert*{B}$. Since this block is sensitive, we know that $z-\pVert*{B}<a_z$ i.e., $\pVert*{B}>(z-a_z)$. We can now define a new block $B_{new}$ such that $B_{new}\subseteq B$ and $\pVert*{B_{new}}=(z-a_z+1)$. Clearly, $B_{new}$ is also sensitive. So, if we have $\bs(f,z)$ disjoint sensitive blocks for the input, we can replace each block $B$ of $k$ 1s with a block $B_{new}$ of $(z-a_z+1)$ 1s.  Similarly, we can also replace any block $C$ of $k$ 0s with a block $C_{new}$ of $(b_z-z+1)$ 0s. 

This combined with Lemma~\ref{oneTypeBlocksOnly} tells us that for an input of Hamming weight $z$ at $[a_z..z..b_z]$ being the largest contiguous block of Hamming weights including $z$ where the function value remains the same, we can define $\bs(f,z)$ disjoint sensitive blocks such that every block contains either $(z-a_z+1)$ indices where $x_i=1$ or $(b_z-z+1)$ indices where $x_i=0$. It is necessary to note that if $a_z=0$, no sensitive block having just 1s is possible. Similarly, if $b_z=n$, no sensitive block having just 0s is possible,

So, the block sensitivity can be written as $\bs(f,z) = \floor*{\frac{z}{z-a_z+1}}+\floor*{\frac{n-z}{b_z-z+1}}$. It is useful to note that when $a=0$, $\bs(f,z) = \floor*{\frac{n-z}{b_z-z+1}}\le \s(f,b_z)$ and when $b_z=n$, $\bs(f,z) = \floor*{\frac{z}{z-a_z+1}} \le \s(f,a_z)$.

\end{proof}

    Using Lemma~\ref{blockSensitivityFormula}, we restrict the class of functions where $\bs(f)/\s(f)$ obtains its maximum value.
    \begin{lemma}
    Suppose $f$ is a symmetric Boolean function such that $\bs(f) > \s(f)$. There exist  $f', a, b$ where $a \geq 2, b \leq n-2$ such that
    \begin{itemize}
        \item $f'(z) = 1$ iff $a \leq z \leq b$, and 
        \item $\frac{\bs(f')}{\s(f')} \geq \frac{\bs(f)}{\s(f)} $.
    \end{itemize} 
    \end{lemma}
    
    \begin{proof}
    Since $\bs(f) > \s(f)$, we can assume that $f$ satisfies the following three properties.
    \begin{enumerate}
        \item It has an input with Hamming weight $z$ such that $a_z \neq 0$ as well as $b_z \neq n$. If there was no $z$ with $a_z \neq 0$ and with $b_z \neq n$, then for any $z \in [n]$, $\bs(f,z)$ is either $\left \lfloor \frac{n-z}{b_z-z+1} \right \rfloor$ or $\left \lfloor \frac{z}{z-a_z+1} \right \rfloor$ as $f$ is not a constant function. As seen in the proof for Lemma~\ref{blockSensitivityFormula}, these values are at max $\s(f,b)$ and $\s(f,a)$ respectively. Thus $\bs(f) = \s(f)$.
        \item It does not have any inputs with Hamming weight $z$ such that $a_z$ = $b_z$. For any input with Hamming weight $z$ such that  $a_z$ = $b_z$, $\s(f,z) = n$. Therefore, $\bs(f) = \s(f)$.
        \item It does not have any inputs with Hamming weight $z$ such that $a_z = 1$ or $b_z = n-1$. If we have $a_z = 1$, then $\s(f) = \s(f,0) = n$. Similarly if $b_z = n-1$, then $\s(f) = \s(f,n) = n$. Thus $\bs(f) = \s(f)$.  
    \end{enumerate}
    Given an $f$ satisfying the above three properties, let $z'$ be the Hamming weight where block sensitivity is achieved, i.e., $\bs(f) = \bs(f,z')$. Then the required $f'$ for the Lemma statement can be defined as follows:
    \begin{equation*}
        f'(z) = \begin{cases}
        1, & \textit{ if $z \in [a_{z'}, b_{z'}]$} \\
        0, & \textit{otherwise.} \\
        \end{cases}
    \end{equation*}
    
    We know that $a_{z'} \geq 2$ and $b_{z'} \leq n-2$.    
    From Lemma~\ref{blockSensitivityFormula}, we can say that $\bs(f,z') = \bs(f',z')$. 
    So $\bs(f) = \bs(f,z') = \bs(f',z') \leq \bs(f')$ (actually the last inequality is equality because we have constructed our function $f'$ in such a way that this $z'$ will give us the maximum block sensitivity). 
    
    Similarly, the function $f'$ is sensitive at only 4 points, hence 
    \begin{center}
      $\s(f')$ = max($\s(f,a_{z'}-1)$, $\s(f,a_{z'})$, $\s(f,b_{z'})$, $\s(f,b_{z'}+1)$).   
    \end{center}
    
    Thus, we can say that $\s(f') \leq \s(f)$ (on these aforementioned points sensitivity for both functions is same, though $f$ may have other sensitive points). 
    
    In case of functions with multiple Hamming weights which give maximum $\bs(f)$, just consider the Hamming weight which gives us a smaller $\s(f')$ value. 
    
    Thus $\frac{\bs(f')}{\s(f')} \geq \frac{\bs(f)}{\s(f)}$, proving the Lemma. 
    \end{proof}
    
We saw that the ratio of $\bs(f)$ and $\s(f)$ is maximized by a function which is non-zero only in a contiguous block. We further restrict the optimal function and claim that the length of this contiguous block (where the function is non-zero) should be $2$.

\begin{lemma}
\label{maxSep}
Let $a < b$ be two integers such that $a\geq 2$ and $b \leq n-2$. Consider functions of the following form,
    
\begin{equation*}
f(z)=\begin{cases}
1, & \text{if $a \leq z \leq b$ }\\
0, & \textit{otherwise}.
\end{cases}
\end{equation*}

The separation $\frac{\bs(f)}{\s(f)}$ is maximum when $b-a = 1$. 
\end{lemma}

    \begin{proof}
    For any function in the above form, $\s(f)$ is $\max(n-a+1, n-b, a, b+1)$. Since $a<b$, we can further restrict $\s(f)$ to be $\max(n-a+1, b+1)$.  Since $f'(z) = f(n-z)$ has the same block sensitivity and sensitivity, we can assume $ n \leq  a + b$. So $\s(f) = b+1$. 
    
    
    
  From Lemma~\ref{blockSensitivityFormula}, the maximum $\bs(f)$ will be achieved when $z \in [a,b]$, hence our task is to maximise the ratio 
  \begin{equation*}
      \frac{\bs(f)}{\s(f)} = \frac{\left \lfloor \frac{z}{z-a+1} \right \rfloor + \left \lfloor \frac{n-z}{b-z+1} \right \rfloor}{b+1},
  \end{equation*}
  when $z\in [a,b]$.

 
 Let $z'$ be the Hamming weight where block sensitivity for function $f$ is achieved. We know that $z' \in [a,b]$. We construct a new function $f'$ which is non-zero only in a contiguous block between $a'$ and $b'$, and $b' -a' = 1$. 
 \begin{itemize}
     \item If $z' = a$ : 
         \begin{equation*}
         f'(z)=\begin{cases}
         1, & \text{if $a' = a \leq z \leq a+1 = b'$ }\\
         0, & \textit{otherwise}.
        \end{cases}
        \end{equation*}
     
     \item If $z' > a$ :
        \begin{equation*}
        f'(z)=\begin{cases}
        1, & \text{if $a' = a-1 \leq z \leq a = b'$ }\\
        0, & \textit{otherwise}.
        \end{cases}
        \end{equation*}
 \end{itemize}
 
Since $a'\geq a$ and $b' \leq b$ in both cases, $\s(f) \geq \s(f')$ and $\bs(f) \leq \bs(f')$, implying $\bs(f')/\s(f') \geq \bs(f)/\s(f)$. 




Since for every function $f$ which is non-zero only in a contiguous block between $a$ and $b$, there exists a function $f'$, as defined above, with a better $\bs(f)/\s(f)$ ratio, we can say that the separation $\bs(f)/\s(f)$ is maximum when $b-a = 1$.

\end{proof}
    
Now, it is easy to see the proof of Theorem~\ref{bsfFinalTheorem} and construct a function which achieves the maximum possible separation.
    
\begin{proof}[Proof of Theorem~\ref{bsfFinalTheorem}]
We know from Lemma~\ref{maxSep} that the maximum separation for $\bs(f)/\s(f)$ happens when $b=a+1$. Thus we can restrict our attention to only functions of the form:
\begin{equation*}
    f(z)=\begin{cases}
    1, & \text{if $a \leq z \leq a+1$ } \\
    0, & \text{otherwise.}
  \end{cases}
\end{equation*}
Notice that due to symmetry, it suffices to look at $n/2 \leq a \leq n-3$.
We see that the maximum gap between $\bs(f)$ and $\s(f)$ happens when $a=n/2$. Define $G$ to be the function,
\begin{equation*}
    G(z)=\begin{cases}
    1, & \text{ $n/2 \leq z \leq n/2+1$ } \\
    0, & \text{otherwise.}
  \end{cases}
\end{equation*}

Notice that $\bs(G) = 3n/4$ and $\s(G) = n/2+2$. Since this is the function with best separation between $\bs$ and $\s$, we get 
\[ \frac{\bs(f)}{\s(f)} \leq 3/2 \]
for any total symmetric function $f$. Clearly this bound is tight, because $G$ achieves this factor.
\end{proof}


\subsection{Bounds on \texorpdfstring{$\bs$}{}, \texorpdfstring{$\RC$}{} and \texorpdfstring{$\MM$}{} for Gap Majority}
\label{sec:gapMajappendix}

\begin{lemma}
\label{bsGapMaj}
$\bs(\gm_n) = \theta(\sqrt{n})$. 
\end{lemma}
\begin{proof} 
We know that the Hamming weights on which $\gm_n$ is defined differ in their Hamming weights by exactly $2\sqrt{n}$ [Definition~\ref{defGM}]. To compute $\bs_{0}(\gm_n)$, each block comprises of $2\sqrt{n}$ zeroes which are flipped to ones. The number of distinct blocks that are then possible are $(n/2 + \sqrt{n})/2\sqrt{n}$ which is $\theta(\sqrt{n})$. Similar argument holds true for $\bs_{1}(\gm_n)$. 
 Thus $\bs(\gm_n)$ = max$\{\bs_{0}(\gm_n), \bs_{1}(\gm_n)\}$ = $\theta(\sqrt{n})$.
\end{proof}

\begin{lemma}
\label{RCGapMaj}
$\RC(\gm_n) = \theta(\sqrt{n})$. 
\end{lemma}
\begin{proof} 

From~\cite{tal_properties_2013}, we know that $\RC(f) = \Theta(\FC(f))$ and we also know that $\FC(f) = max_{x \in \text{Dom(f)}} \FC(f,x)$

The linear program for $\FC(f,x)$ is:
\begin{center}
minimize $\Sigma_i z_{i}$ 
\\ subject to $ \sum_{i : x_{i} \neq y_{i}} z_{i} \geq 1 , \forall  y $ :  $f(x) \neq f(y)$ 
\\$z_{i} \in [0,1]$.
\end{center}

Using the linear program for $\FC(f,x)$, for any input $x$, $\forall i \in [n]$ we assign $z_{i}$ as $1/\sqrt{n}$. Since any two inputs $x$ and $y$, where $\gm_n(x)$ and $\gm_n(y)$ are different, differ in at least $2\sqrt{n}$ bits, the condition is satisfied and this assignment becomes a feasible assignment. The value of the objective function then becomes $\sqrt{n}$. Since it is a minimization program, we know that for any $x$, $\FC(\gm_n,x) \leq \sqrt{n}$ and hence $\FC(\gm_n) \leq \sqrt{n}$. \newline
This implies that $\RC(\gm_n) = O(\sqrt{n})$. From~\cite{tal_properties_2013}, we also know that $ \forall f, \bs(f) \leq \RC(f)$. Combining this with the Lemma~\ref{bsGapMaj}, we get that $\RC(\gm_n) = \Omega(\sqrt{n})$. Thus, $\RC(\gm_n) = \theta(\sqrt{n})$. 

\end{proof}

\begin{lemma}
\label{MMfGapMaj}
$\Adv(\gm_n) = \theta(\sqrt{n})$
\end{lemma}
\begin{proof} 
 
 As seen in Definition~\ref{MMfDefinition} of  $\MM(f)$, assigning $w(x,i)$ to be $1/\sqrt{n}$ for all inputs $x$ and all indices $i$ is a feasible solution for $\MM(\gm_n)$. This is because the number of indices $i$ where the inputs $x$ and $y$ such that $f(x) \neq f(y)$ differ will be at least $2\sqrt{n}$ thus always satisfying the condition. The objective value for this weight scheme $w$ then becomes $\sqrt{n}$. Since $\MM(f)$ is a minimization program over all such feasible weight functions, $\MM(\gm_n) \leq \Theta(\sqrt{n})$. Thus, $\Adv(\gm_n) = O(\sqrt{n})$. 
 Combining this with Lemma~\ref{lowerBoundGapMaj}, we get $\Adv(\gm_n) = \theta(\sqrt{n})$. 

\end{proof}

\end{document}